\documentclass[11pt]{article}

\usepackage[margin=20mm]{geometry}

\usepackage{amssymb}  % for \mathbb
\usepackage{amsthm}  % for \theoremstyle
\usepackage{bigstrut}  % for tabular row height correction after hline
\usepackage{enumitem}  % for enumerate, itemize options
\usepackage{graphicx}  % for \includegraphics
\usepackage{mathtools}  % for \boldsymbol
\usepackage{multirow}  % for \multirow

\usepackage{natbib}

\makeatletter\def\@seccntformat#1{\csname the#1\endcsname.\ \,}\makeatother
\usepackage{url}
\usepackage[pagebackref,colorlinks,bookmarks,allcolors=blue]{hyperref}

\newcommand*{\doi}[1]{doi:\href{https://doi.org/#1}{\nolinkurl{#1}}}

\newcommand{\myorcidshort}[1]{\href{#1}{\includegraphics[height=1.9ex]{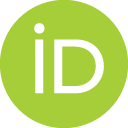}}}
\newcommand{\myorcidlong}[1]{\myorcidshort{#1} \url{#1}}

\usepackage[nameinlink]{cleveref}

\newtheorem{theorem}{Theorem}[section]
\newtheorem{lemma}[theorem]{Lemma}
\theoremstyle{definition}
\newtheorem{remark}[theorem]{Remark}
\newtheorem{example}[theorem]{Example}

\newlist{enumerateroman}{enumerate}{3}
\setlist[enumerateroman]{wide=0pt,label={\upshape(\roman*)}}

\newcommand{\MS}{\boldsymbol{M\!S}}
\newcommand{\myssa}{\sum_{i} \alpha_i^2}
\newcommand{\myssaa}{\sum_{i,j} \alpha_{i(j)}^2}
\newcommand{\myssaaa}{\sum_{i,j,k} \alpha_{i(jk)}^2}

\begin{document}

\title{\vspace{-1.5cm}\textbf{Minimal sample size
in balanced ANOVA models of crossed, nested, and mixed classifications}}

\renewcommand{\thefootnote}{$*$}
\footnotetext[1]{Corresponding author.}

\renewcommand{\thefootnote}{}
\footnotetext[2]{\textit{E-mail address:}
\url{bernhard.spangl@boku.ac.at} (B.~Spangl).
\textit{Postal address:}
Institute of Statistics, University of Natural Resources and Life Sciences, Vienna,
Gregor-Mendel-Straße 33, 1180 Vienna, Austria.}

\date{\vspace{-1.5cm}}

\maketitle

\begin{center}  \large
Bernhard Spangl$^{\text{a},*}$
\myorcidshort{https://orcid.org/0000-0001-6222-2408},
Norbert Kaiblinger$^{\text{b}}$
\myorcidshort{https://orcid.org/0000-0001-6280-5929},
Peter Ruckdeschel$^{\text{c}}$
\myorcidshort{https://orcid.org/0000-0001-7815-4809},
Dieter Rasch$^{\text{a}}$
\myorcidshort{https://orcid.org/0000-0001-9324-9910}
\end{center}

\begin{center}  \footnotesize
$^\text{a}${Institute of Statistics,
University of Natural Resources and Life Sciences, Vienna, Austria} \\
$^\text{b}${Institute of Mathematics,
University of Natural Resources and Life Sciences, Vienna, Austria} \\
$^\text{c}${Institute for Mathematics,
Carl von Ossietzky University Oldenburg, Germany}
\end{center}

\medskip\hrule\medskip

\begin{small}

\noindent\textbf{Abstract:}\quad
We consider balanced one-, two- and three-way ANOVA models
to test the hypothesis that the fixed factor $A$ has no effect.
The other factors are fixed or random.
We determine the noncentrality parameter for the exact $F$\mbox-test,
describe its minimal value by a sharp lower bound,
and thus we can guarantee the worst case power for the $F$\mbox-test.
These results allow us to compute the minimal sample size,
i.e., the minimal number of experiments needed.
We also provide a structural result for the minimum sample size,
proving a conjecture on the optimal experimental design.

\medskip

Power. Experimental size determination.

\medskip

\noindent\textbf{MSC\,2010:}\quad
62K; 62J.

\end{small}
\smallskip
\hrule

\section{Introduction}

Consider a balanced one-, two- or three-way ANOVA model
with fixed factor $A$
to test the null hypothesis~$H_0$ that $A$ has no effect,
that is, all levels of $A$ have the same effect.
The other factors are denoted $B,C$ (crossed with or nested in $A$)
or $U,V$ (factors that $A$ is nested in).
They can be fixed factors (printed in normal font)
or random factors (printed in bold).
As usual in ANOVA we assume identifiability, normality, independence,
homogeneity, and compound symmetry
\citep{MAXWELL2017,SCHEFFE1959}.
In particular, the fixed effects are identifiable
and the random effects and errors have a normal distribution
with mean zero and they are mutually independent.
By $A\times B$ we denote crossed factors with interaction,
by $A \succ B$ we denote that $B$ is nested in $A$.
Practical examples that are modeled by crossed, nested and mixed classifications
are included, for example, in
\citet{CANAVOS2009},
\citet{DONCASTER2007},
\citet{JIANG2007},
\citet{MONTGOMERY2012},
\citet{RASCH1971},
\citet{RASCH2011},
\citet{RASCH2012},
\citet{RASCH2018},
\citet{RASCH2020b}.
The number of levels of $A$ ($B$, $C$, $U$, $V$)
is denoted by $a$ ($b$, $c$, $u$, $v$, respectively).
The effects are denoted by Greek letters.
For example, the effects of the fixed factor $A$
in the one-way model $A$,
the two-way nested model $V \succ A$,
and the three-way nested model $U \succ V \succ A$ read
\begin{equation}  \label{eq:maineffects}
\alpha_i, \ \alpha_{i(j)}, \ \alpha_{i(jk)},
\qquad
i=1,\dots,a,
\quad
j=1,\dots,v,
\quad
k=1,\dots,u ~.
\end{equation}
The numbers of levels (excluding $a$) and the number of
replicates $n$ will be called parameters in this article.

This article derives the details for the noncentrality parameter
and we show how to obtain the minimum sample size for a large family of ANOVA models.

\begin{itemize}

\item
We derive the details for the noncentrality parameter (\Cref{thm:main}).

\item
We derive the worst case noncentrality parameter (\Cref{thm:ineq}),
required to obtain the guaranteed power of an ANOVA experiment.

\item
We show how to determine the minimal experimental size for ANOVA experiments
by a new structural result that we call ``pivot'' effect (\Cref{thm:pivot}).
The ``pivot'' effect means one of the parameters
(the ``pivot'' parameter) is more power-effective than the others.
Considering this ``pivot'' effect is not only helpful for planning experiments
but is indeed necessary in certain models, see \Cref{rema:main}\ref{item:fullpower}.
\end{itemize}

Our main results are thus for the exact $F$\mbox-test noncentrality parameter,
the power, and the minimum sample size determination,
see \Cref{sec:main}.
In \Cref{sec:approx} we include two exceptional models
that do not have an exact $F$\mbox-test.
In \Cref{sec:integer} we discuss the distinction between
real and integer parameters for some of our results.
The proofs are in \mbox{\ref{sec:proofs}}.

\section{Main results}  \label{sec:main}

Consider a balanced 1-, 2- or 3-way ANOVA model,
with the notation above,
to test the null hypothesis $H_0$ that the fixed factor $A$ has no effect.
For most of these models an exact $F$\mbox-test exists,
under the usual assumptions mentioned above.
The test statistic $\boldsymbol{F}_{A}$ is given by a ratio
whose numerator is given by the mean squares (MS) of the fixed factor~$A$,
denoted by $\MS_{A}$. The denominator depends on the model.
The respective test statistic has an $F$\mbox-distribution
(central under $H_{0}$, noncentral in general).
We denote its parameters by
the numerator d.f.\ $\mathit{df_1}$, the denominator d.f.\ $\mathit{df_2}$,
and the noncentrality parameter $\lambda$.
The notation d.f.\ is short for degrees of freedom.

By $\sigma_y^2$ we denote the total variance, it is the sum of the
variance components,
such as $\sigma_{\beta}^2$ (the variance component of the factor $B$)
and the error term variance $\sigma^2$.

\subsection{The noncentrality parameter}  \label{sec:noncent}

Our first main result lists $\mathit{df_1}$, $\mathit{df_2}$
and the exact form of the noncentrality parameter $\lambda$.
Our expressions for $\lambda$
show the detailed form in which the variance components occur.
This exact form of $\lambda$ is the key to a reliable power analysis,
which is essential for the design of experiments.

\begin{theorem}  \label{thm:main}
Consider a balanced 1-, 2- or 3-way ANOVA model, with
the assumptions of identifiability, normality, independence,
homogeneity, and compound symmetry.
We test the null hypothesis $H_0$ that the fixed factor $A$ has no effect.
Then, under the assumption
that an exact $F$\mbox-test exists, the test statistic has an $F$\mbox-distribution
(central under $H_{0}$, noncentral in general)
with numerator d.f.\ $\mathit{df_1}$, denominator d.f.\ $\mathit{df_2}$,
and noncentrality parameter $\lambda = R S/T$
obtained from {\upshape\Cref{tab:main}}.
\end{theorem}

The proof of \Cref{thm:main} is in \ref{sec:proofs}.

\begin{table}[!htbp]
\caption{
List of 1-, 2- and 3-way ANOVA models with fixed factor $A$,
for use in \Cref{thm:main} etc.
The letters $a,b,\dots$ denote the numbers of levels,
and $n$ is the number of replicates.
To point out equivalences, the variance component notation is simplified,
such as $\sigma_{\alpha,\beta}^2$ represents both
$\sigma_{\alpha\beta}^2$ and $\sigma_{\beta(\alpha)}^2$.
In the first column, bold font indicates random factors.
The ``pivot'' parameter, also printed in bold to indicate randomness,
is the most power-effective parameter, see \Cref{thm:pivot}.
}
\label{tab:main}

\centering

\resizebox*{!}{.9\textheight}{
\begin{tabular}{|l|c|ll|lll|}
\hline
\multirow{2}{*}{Model}
& Pivot pa-
& \multirow{2}{*}{$\mathit{df_1}$} & \multirow{2}{*}{$\mathit{df_2}$}
& \multicolumn{3}{l|}{$\lambda = R S/T$
} \\
\cline{5-7}
& rameter & & & $R$ & $S$ & $T$ \bigstrut[t] \\
\hline
\hline
$A$ & $\boldsymbol{n}$ & $a-1$ & $a (n-1)$
& $n$ & $\myssa$ & $\sigma^2$ \bigstrut[t] \\
\hline
\hline
$A \times B$ & $\boldsymbol{n}$ & `` & $ab (n-1)$ & $b n$ & `` & `` \\
$A \succ B$ & `` & `` & `` & `` & `` & `` \\
\hline
$A \times \boldsymbol{B}$ & $\boldsymbol{b}$ & `` & $(a-1) (b-1)$
& $b$ & `` & $ \sigma_{\alpha,\beta}^2 + \frac{1}{n}\sigma^2$ \bigstrut[t] \\
$A \succ \boldsymbol{B}$ & `` & `` & $a (b-1)$ & `` & `` & `` \\
\hline
$V \succ A$ & $\boldsymbol{n}$ & $v(a-1)$ & $v a (n-1)$
& $n$ & $\myssaa$ & $\sigma^2$ \bigstrut[t] \\
$\boldsymbol{V} \succ A$ & `` & `` & ``
& `` & `` & `` \\
\hline
\hline
$A \times B \times C$ & $\boldsymbol{n}$ & $a-1$ & $abc (n-1)$
& $b c n$ & $\myssa$ & $\sigma^2$ \bigstrut[t] \\
$A \succ B \succ C$ & `` & `` & `` & `` & `` & `` \\
$(A \times B) \succ C$ & `` & `` & `` & `` & `` & `` \\
$(A \succ B) \times C$ & `` & `` & `` & `` & `` & `` \\
$A \times (B \succ C)$ & `` & `` & `` & `` & `` & `` \\
\hline
$A \succ B \succ \boldsymbol{C}$ & $\boldsymbol{c}$ & `` & $a b(c-1)$
& $b c$ & `` & $ \sigma_{\alpha,\beta,\gamma}^2 + \frac{1}{n}\sigma^2$ \bigstrut[t] \\
$(A \times B) \succ \boldsymbol{C}$ & `` & `` & `` & `` & `` & `` \\
$A \times (B \succ \boldsymbol{C})$ & `` & `` & $(a-1) b (c-1)$
& `` & `` & `` \\
$(A \succ B) \times \boldsymbol{C}$ & `` & `` & $(a-1)(c-1)$
& $c$ & `` & $ \sigma_{\alpha,\gamma}^2 + \frac{1}{bn}\sigma^2$ \\
\hline
$A \times \boldsymbol{B} \times C$ & $\boldsymbol{b}$ & `` & $(a-1)(b-1)$
& $b$ & `` & $ \sigma_{\alpha,\beta}^2 + \frac{1}{cn}\sigma^2$ \bigstrut[t] \\
$(A \times \boldsymbol{B}) \succ C$ & `` & `` & `` & `` & `` & `` \\
$A \times (\boldsymbol{B} \succ C)$ & `` & `` & `` & `` & `` & `` \\
$A \succ \boldsymbol{B} \succ C$ & `` & `` & $a (b-1)$ & `` & `` & `` \\
$(A \succ \boldsymbol{B}) \times C$ & `` & `` & `` & `` & `` & `` \\
\hline
$A \succ \boldsymbol{B} \succ \boldsymbol{C}$ & $\boldsymbol{b}$
& `` & `` & `` & `` & $ \sigma_{\alpha,\beta}^2 + \frac{1}{c} \sigma_{\alpha,\beta,\gamma}^2 + \frac{1}{cn}\sigma^2$ \bigstrut[t] \\
$(A \times \boldsymbol{B}) \succ \boldsymbol{C}$
& `` & `` & $(a-1) (b-1)$ & `` & `` & `` \\
$A \times (\boldsymbol{B} \succ \boldsymbol{C})$ & `` & `` & `` & `` & `` & `` \\
\hline
\hline
$V \succ A \succ B$ & $\boldsymbol{n}$ & $v (a-1)$ & $vab (n-1)$
& $b n$ & $\myssaa$ & $\sigma^2$ \bigstrut[t] \\
$(V \succ A) \times B$ & `` & `` & `` & `` & `` & `` \\
$\boldsymbol{V} \succ A \succ B$ & `` & `` & `` & `` & `` & `` \\
$(\boldsymbol{V} \succ A) \times B$ & `` & `` & `` & `` & `` & `` \\
\hline
$V \succ A \succ \boldsymbol{B}$ & $\boldsymbol{b}$ & `` & $v a (b-1)$
& $b$ & `` & $ \sigma_{\nu,\alpha,\beta}^2 + \frac{1}{n}\sigma^2$ \bigstrut[t] \\
$\boldsymbol{V} \succ A \succ \boldsymbol{B}$ & `` & `` & `` & `` & `` & `` \\
$(V \succ A) \times \boldsymbol{B}$ & `` & `` & $v (a-1) (b-1)$ & `` & `` & `` \\
$(\boldsymbol{V} \succ A) \times \boldsymbol{B}$
& `` & `` & `` & `` & `` & `` \\
\hline
\hline
$U \succ V \succ A$ & $\boldsymbol{n}$ & $u v (a- 1)$ & $uva (n-1)$
& $n$ & $\myssaaa$ & $\sigma^2$ \bigstrut[t] \\
$(U \times V) \succ A$ & `` & `` & `` & `` & `` & `` \\
$\boldsymbol{U} \succ V \succ A$ & `` & `` & `` & `` & `` & `` \\
$U \succ \boldsymbol{V} \succ A$ & `` & `` & `` & `` & `` & `` \\
$(U \times \boldsymbol{V}) \succ A$ & `` & `` & `` & `` & `` & `` \\
$\boldsymbol{U} \succ \boldsymbol{V} \succ A$ & `` & `` & `` & `` & `` & `` \\
$(\boldsymbol{U} \times \boldsymbol{V}) \succ A$ & `` & `` & `` & `` & `` & `` \\
\hline

\end{tabular}
}
\end{table}

\begin{example}  \label{exam:denominator}
For the model $A \succ \boldsymbol{B} \succ \boldsymbol{C}$,
\Cref{thm:main} states that
the test statistic $\boldsymbol{F}_{A} = \MS_{A} / \MS_{B \, \text{in} \, A}$
has an $F$\mbox-distribution
(central under $H_0$, noncentral in general)
with numerator d.f.\ $\mathit{df_1}=a-1$, denominator d.f.\ $\mathit{df_2}=a(b-1)$,
and noncentrality parameter
\begin{equation*}
\lambda = R S/T
= b \cdot \frac{\myssa}{\sigma_{\beta(\alpha)}^2 + \frac{1}{c} \sigma_{\gamma(\alpha\beta)}^2 + \frac{1}{cn}\sigma^2} ~.
\end{equation*}
\end{example}

\begin{remark}  \label{rema:main}
\begin{enumerateroman}
\item
The models $A \times \boldsymbol{B} \times \boldsymbol{C}$
and $(A \succ \boldsymbol{B}) \times \boldsymbol{C}$
are excluded from \Cref{tab:main},
since an exact $F$\mbox-test does not exist, see \Cref{sec:approx}.
We also exclude the nesting of crossed factors into others, such as
$A \succ (B \times C)$.

\item  \label{item:fullpower}
From inspecting the expression for $\lambda$ in \Cref{exam:denominator}
we obtain the following somewhat surprising observation.
If $n$ increases, then clearly $\lambda$ increases,
but in the limit $n\to\infty$ we do not obtain $\lambda\to\infty$.
It implies that increasing the number of replicates $n$
increases the power but there is a limit for the power
if only $n$ is increased.
This observation affects each model in \Cref{tab:main} with $T$
consisting of more than one term.
These are exactly the models which in \Cref{tab:main} do not have
the parameter $n$ in the ``pivot'' column.
In fact, the ``pivot'' effect (\Cref{thm:pivot} below) shows that
for these models not $n$ but a different parameter should be increased
to achieve any given prespecified power.
\end{enumerateroman}
\end{remark}

\subsection{Least favorable case noncentrality parameter}

For an exact $F$\mbox-test, the computation of the power is
immediate:
Given the type~I risk $\alpha$, obtain the type~II risk $\beta$
by solving
\begin{equation}  \label{eq:pow}
F_{\mathit{df_1}, \mathit{df_2}; 1 - \alpha} = F_{\mathit{df_1}, \mathit{df_2}; \beta}^\lambda ~, 
\end{equation}
where $F_{\nu_1, \nu_2; \gamma}^\lambda$ denotes the $\gamma$-quantile of
the $F$-distribution with degrees of
freedom $\nu_1$ and $\nu_2$ and noncentrality parameter $\lambda$.
Then $P = 1-\beta$ is the power of the test.
The next theorem is our second main result,
we determine the noncentrality parameter $\lambda_{\min}$
in the least favorable case, that is,
the sharp lower bound in $\lambda \geq \lambda_{\min}$.
Using $\lambda_{\min}$ in \eqref{eq:pow}
yields the guaranteed power $P_{\min} = (1-\beta)_{\min}$ of the test.

Let $\delta$ denote the minimum difference to be detected between the
smallest and the largest treatment effects,
i.e., between the minimum $\alpha_{\min}$ and the maximum $\alpha_{\max}$
of the set of the main effects of the fixed factor~$A$,
\begin{equation}  \label{eq:delta}
\delta = \alpha_{\max} - \alpha_{\min} ~.
\end{equation}
We assume the standard condition
to ensure identifiability of parameters,
which is that $\alpha$ has zero mean in all directions
\citep[pp.\,157, 169, 178]{FOX2015},
\citep[Sec.\,3.3.1.1]{RASCH2011},
\citep[Sec.\,5]{RASCH2018},
\citep[Sec.\,5]{RASCH2020b},
\citep[Sec.\,4.1, p.\,92]{SCHEFFE1959},
\citep[p.\,415, Sec.\,7.2.i]{SEARLE2016}.
That is, exemplified for three models,
\begin{equation}  \label{eq:meanzero}
\begin{aligned}
A & \quad\Rightarrow\quad \sum_{i} \alpha_i^2 = 0, \\
V \succ A & \quad\Rightarrow\quad
\sum_{i} \alpha_{i(j_0)}^2 = 
\sum_{j} \alpha_{i_0(j)}^2 = 0, \text{ for any $i_0,j_0$}, \\
U \succ V \succ A & \quad\Rightarrow\quad
\sum_{i} \alpha_{i(j_0 k_0)}^2 =
\sum_{j} \alpha_{i_0(j k_0)}^2 =
\sum_{k} \alpha_{i_0(j_0 k)}^2 = 0, \text{ for any $i_0, j_0, k_0$} ~.
\end{aligned}
\end{equation}

\begin{theorem}  \label{thm:ineq}
We have the following lower bounds for the noncentrality parameter $\lambda$.

\begin{enumerateroman}
\item  \label{item:rbound}
With the parameter or product of parameters
denoted $R$ in {\upshape\Cref{tab:main}}, we have
\begin{equation*}
\lambda \geq \frac{R}{2} \cdot \frac{\delta^2}{\sigma_{y}^2} .
\end{equation*}

More precisely,
denoting by $\sigma_{y,\mathit{active}}^2 \leq \sigma_{y}^2$ the sum of those variance components
that occur in $T$, we have
\begin{equation*}
\lambda \geq \frac{R}{2} \cdot \frac{\delta^2}{\sigma_{y,\mathit{active}}^2} .
\end{equation*}

\item
For the models in {\upshape\Cref{tab:main}} that involve a factor $V$
that $A$ is nested in, let $m = \max(v,a)$.
Then the lower bound in \ref{item:rbound} can be raised to
\begin{equation*}
\lambda \geq \frac{R}{2} \cdot \frac{\delta^2}{\sigma_{y,\mathit{active}}^2}
\cdot \frac{m}{m-1}.
\end{equation*}
\item
For the models in {\upshape\Cref{tab:main}} that involve the factors $U,V$
that $A$ is nested in, let $m_1 \leq m_2 \leq m_3$ denote $a,u,v$
sorted from least to greatest.
Then the lower bound in \ref{item:rbound} can be raised to
\begin{equation*}
\lambda \geq \frac{R}{2} \cdot \frac{\delta^2}{\sigma_{y,\mathit{active}}^2}
\cdot \frac{m_2 m_3}{(m_2-1)(m_3-1)}.
\end{equation*}
\end{enumerateroman}
\end{theorem}

The proof of \Cref{thm:ineq} is in \ref{sec:proofs}.

\begin{remark}  \label{rema:ineq}
\begin{enumerateroman}

\item
The importance of a lower bound for the noncentrality parameter $\lambda$
is its use for the power analysis,
required for the design of experiments.
By \Cref{thm:ineq} we establish such a bound.
The difference to the previous literature \citep{RASCH2011,RASCH2012}
is that we use the correct,
detailed form of the noncentrality parameter $\lambda$ from \Cref{thm:main},
and we use the new, sharp bound for the sum of squared effects
from \citep{KAIBLINGER2020}.

\item  \label{item:leastfavorable}
The bounds in \Cref{thm:ineq} are sharp.
The extremal case (minimal $\lambda$) occurs
if the main effects \eqref{eq:maineffects} of the factor $A$
are least favorable,
while satisfying \eqref{eq:delta} and \eqref{eq:meanzero},
and also the variance components are least favorable,
while their sum does not exceed $\sigma_{y}^2$.

For the extremal $\alpha_i, \ \alpha_{i(j)}, \ \alpha_{i(jk)}$
configurations we refer to \citet{KAIBLINGER2020}.
The least favorable splitting of $\sigma_{y}^2$
is that the total variance
is consumed entirely by the first term of $T$ in {\upshape\Cref{tab:main}},
see the worst cases in \Cref{exam:ineq}\ref{item:crossedexam},\ref{item:nestedexam}.

\item  \label{item:mostfavorable}
If in a model there are ``inactive'' variance components
(i.e., some components of the model do not occur in $T$),
then the most favorable splitting of $\sigma_{y}^2$ is that
the total variance tends to be consumed entirely by inactive components.
In these cases $\lambda$ goes to infinity, $\lambda\to\infty$.
See the best case in \Cref{exam:ineq}\ref{item:crossedexam}.

If in a model all variance components are ``active''
(i.e., all components of the model also occur in $T$),
then the most favorable splitting of $\sigma_{y}^2$ is that
the total variance is consumed entirely by the last term of $T$.
See the best case in \Cref{exam:ineq}\ref{item:nestedexam}.
\end{enumerateroman}
\end{remark}

\begin{example}  \label{exam:ineq}
\begin{enumerateroman}
\item  \label{item:crossedexam}
For the model $A \times \boldsymbol{B} \times C$,
from \Cref{tab:main} we have
\begin{equation*}
T = \sigma_{\alpha\beta}^2 + \frac{1}{cn}\sigma^2 ~.
\end{equation*}
The ``active'' variance components are defined to be the
variance components that occur in $T$,
\begin{equation*}
\sigma_{y}^2
= {\underbrace{\sigma_{\alpha\beta}^2 + \sigma^2}_{\displaystyle\sigma_{y,\mathit{active}}^2}}
+ \sigma_{\beta}^2 + \sigma_{\beta\gamma}^2 + \sigma_{\alpha\beta\gamma}^2 .
\end{equation*}
Since $R=b$, by \Cref{thm:ineq} we obtain for the noncentrality parameter $\lambda$,
\begin{equation*}
\begin{aligned}
\lambda
\geq \frac{b}{2} \cdot \frac{\delta^2}{\sigma_{y,\mathit{active}}^2}
\geq \frac{b}{2} \cdot \frac{\delta^2}{\sigma_{y}^2} ~.
\end{aligned}
\end{equation*}
Since the first term of $T$ is $\sigma_{\alpha\beta}^2$
and the inactive components are $\sigma_{\beta}^2, \sigma_{\beta\gamma}^2, \sigma_{\alpha\beta\gamma}^2$,
we obtain by \Cref{rema:ineq}
that the extremal total variance $\sigma_{y}^2$ splittings are
\begin{equation*}
(\sigma_{\alpha\beta}^2, \sigma^2, \sigma_{\beta}^2, \sigma_{\beta\gamma}^2, \sigma_{\alpha\beta\gamma}^2) \to
\begin{dcases}
(*,0,0,0,0), & \text{ worst, }
\lambda = \frac{b}{2} \cdot \frac{\delta^2}{\sigma_{y}^2}, \\
(0,0,*,*,*), & \text{ best, }
\lambda \to \infty .
\end{dcases}
\end{equation*}

\item  \label{item:nestedexam}
For the model $A \succ \boldsymbol{B} \succ \boldsymbol{C}$,
from \Cref{tab:main} we have
\begin{equation*}
T = \sigma_{\beta(\alpha)}^2 + \frac{1}{c} \sigma_{\gamma(\alpha\beta)}^2 + \frac{1}{cn}\sigma^2 .
\end{equation*}
All variance components occur in $T$,
thus all variance components are ``active'',
\begin{equation*}
\sigma_{y}^2
= \sigma_{y,\mathit{active}}^2
= \sigma_{\beta(\alpha)}^2 + \sigma_{\gamma(\alpha\beta)}^2 + \sigma^2 .
\end{equation*}
Since $R=b$, by \Cref{thm:ineq} we obtain for the noncentrality parameter $\lambda$,
\begin{equation*}
\begin{aligned}
\lambda
\geq \frac{b}{2} \cdot \frac{\delta^2}{\sigma_{y,\mathit{active}}^2}
= \frac{b}{2} \cdot \frac{\delta^2}{\sigma_{y}^2} ~.
\end{aligned}
\end{equation*}
In this model there are no ``inactive'' variance components,
and by \Cref{rema:ineq} we obtain
\begin{equation*}
(\sigma_{\beta(\alpha)}^2, \sigma_{\gamma(\alpha\beta)}^2, \sigma^{2}) \to
\begin{dcases}
(*,0,0), & \text{ worst, }
\lambda = \frac{b}{2} \cdot \frac{\delta^2}{\sigma_{y}^2}, \\
(0,0,*), & \text{ best, }
\lambda = \frac{bcn}{2} \cdot \frac{\delta^2}{\sigma_{y}^2}.
\end{dcases}
\end{equation*}
\end{enumerateroman}
\end{example}

\subsection{Minimal sample size}

The size of the $F$\mbox-test is the product of the parameters,
for the factors that occur in the model,
including the number $n$ of replications.
For prespecified power requirements $P\geq P_0$,
the minimal sample size can be determined by \Cref{thm:ineq}.
Compute $\lambda_{\min}$
and thus obtain the guaranteed power $P_{\min} = (1-\beta)_{\min}$,
for each set of parameters that belongs to a given size,
increasing the size until the power $P_0$ is reached.

The next theorem is the main structural result of our article.
We show that for given power requirements $P\geq P_0$,
the minimal sample size can be obtained
by varying only one parameter, which we call ``pivot'' parameter,
keeping the other parameters minimal.
We thus prove and generalize suggestions in \citet{RASCH2011},
see \Cref{rema:pivot}\ref{item:conj} below.
Part~\ref{item:pivotone} of the next theorem
describes the key property of the ``pivot'' parameter,
part~\ref{item:power} is an intermediate result,
and part~\ref{item:size} is the minimum sample size result.

\begin{theorem}  \label{thm:pivot}
Denote by ``pivot'' parameter the parameter in the second column of {\upshape\Cref{tab:main}}.
Then the following hold.
\begin{enumerateroman}
\item  \label{item:pivotone}
If a parameter increases, then the power increases most
if it is the ``pivot'' parameter.

\item \label{item:power}
For fixed size, if we allow the parameters to be real numbers,
then the maximal power occurs if the ``pivot'' parameter
varies and the other parameters are minimal.

\item  \label{item:size}
For fixed power, if we allow the parameters to be real numbers,
then the minimum size occurs if the ``pivot'' parameter
varies and the other parameters are minimal.
\end{enumerateroman}
\end{theorem}

The proof of \Cref{thm:pivot} is in \ref{sec:proofs}.

\begin{example}
For the model $A \succ \boldsymbol{B} \succ \boldsymbol{C}$,
we have the following.
For given power requirements $P\geq P_0$,
the minimal sample size is obtained
by varying the parameter $b$, keeping $c$ and $n$ minimal.
For this and two other examples, see \Cref{tab:MinimumSizeExact}.
\end{example}

\begin{remark}  \label{rema:pivot}
\begin{enumerateroman}
\item
The ``pivot'' parameter in \Cref{thm:pivot},
defined in the second column of \Cref{tab:main},
can also be identified directly
from the model formula in the first column of the table.
That is, the ``pivot'' parameter is the number of levels
of the random factor nearest to $A$,
if we include the number $n$ of replicates as a virtual random factor,
and exclude factors that $A$ is nested in (labeled $U,V$).
For example, in $A \succ \boldsymbol{B} \succ \boldsymbol{C}$
the random factor $B$ is nearer to $A$ than the random factor $C$
or the virtual random factor of replicates;
and indeed the ``pivot'' parameter is $b$.
Inspired by related comments in \citet[p.\,23]{DONCASTER2007}
we interpret this heuristic observation
as a correlation between higher power effect and higher organizatorial level.

\item  \label{item:conj}
In \citet[p.\,73]{RASCH2011} it is observed that for the two-way model $A \times \boldsymbol{B}$
only the parameter $b$ should vary, but $n$ should be chosen as small as possible,
to achieve the minimum sample size.
For the model $\boldsymbol{V} \succ A$,
it is conjectured \citep[p.\,78]{RASCH2011} that only $n$ should vary,
but $v$ should be as small as possible, to achieve the minimal sample size.
These suggestions are motivated by inspecting the effect of the parameters
on the denominator d.f.\ $\mathit{df_2}$.
By \Cref{thm:pivot}\ref{item:size} we prove the conjecture and generalize these observations.
In fact, from \Cref{tab:main} the ``pivot'' parameter for $A \times \boldsymbol{B}$ is $b$,
and the ``pivot'' parameter for $\boldsymbol{V} \succ A$ is $n$.
Our proof works by inspecting the effect of the parameters not only on $\mathit{df_1}$ and $\mathit{df_2}$
but also on the noncentrality parameter~$\lambda$.
Note we assume that the parameters are real numbers,
for the subtleties of the transition to integer parameters see \Cref{sec:integer}.
\end{enumerateroman}
\end{remark}

\begin{table}
\caption{
Exemplifying the ``pivot'' effect (\Cref{thm:pivot}) for three models.
For each model,
the left table illustrates \Cref{thm:pivot}\ref{item:power}
and the right table illustrates \Cref{thm:pivot}\ref{item:size}.
In the left table for all parameter sets
with prespecified product $bcn=24$ ($vn=12$, respectively),
thus fixed sample size,
the noncentrality parameter~$\lambda$ and the power~$P$
are calculated, sorted by increasing power.
In the right table, for each of four power requirements
($P \geq 0.80,0.85, 0.90, 0.95$),
the parameter set with minimal sample size is calculated.
The parameters are the numbers $b, c, v$ of levels of the random factors
$B$, $C$, $V$, respectively,
and the number $n$ of replicates.
The number of levels of the fixed factor $A$ is $a=6$,
the minimum difference to be detected between the
smallest and the largest treatment effects is $\delta=1$,
and $\alpha = 0.05$.
The variance components are
$(\sigma_{\beta(\alpha)}^2,\ \sigma_{\gamma(\alpha\beta)}^2,\ \sigma^2)$
$= (1/18,1/9,1/6)$,
$(\sigma_{\alpha\gamma}^2,\ \sigma_{\beta(\alpha\gamma)}^2,\ \sigma^2,\ \sigma_{\gamma}^2)$
$= (1/18,1/9,1/6,*)$
and
$(\sigma^2,\ \sigma_{\nu}^2)$
$= (1/4,*)$, respectively.
Here, an asterisk indicates an arbitrary value
since the component is inactive, cf.\ \Cref{rema:ineq}\ref{item:mostfavorable}.}
\label{tab:MinimumSizeExact}

\begin{equation*}
\begin{array}{ll}
\multicolumn{2}{l}{\bullet\ \text{Model $A \succ \boldsymbol{B} \succ \boldsymbol{C}$,
pivot $\boldsymbol{b}$}}
\\

\begin{array}[t]{lllll}
\hline
(\boldsymbol{b},c,n) & df_1 & df_2 & \lambda & P \\
\hline
 (\boldsymbol{2},2,6) & 5 & 6 & 8. & 0.271516 \\
 (\boldsymbol{2},3,4) & 5 & 6 & 9.3913 & 0.314513 \\
 (\boldsymbol{2},4,3) & 5 & 6 & 10.2857 & 0.342042 \\
 (\boldsymbol{2},6,2) & 5 & 6 & 11.3684 & 0.375051 \\
 (\boldsymbol{3},2,4) & 5 & 12 & 11.3684 & 0.527472 \\
 (\boldsymbol{3},4,2) & 5 & 12 & 14.4 & 0.642402 \\
 (\boldsymbol{4},2,3) & 5 & 18 & 14.4 & 0.712478 \\
 (\boldsymbol{4},3,2) & 5 & 18 & 16.6154 & 0.781856 \\
 (\boldsymbol{6},2,2) & 5 & 30 & 19.6364 & 0.897849 \\
\hline
\end{array}
&
\begin{array}[t]{llllll}
\hline
P_{\text{requ.}} & (\boldsymbol{b},c,n) & df_1 & df_2 & \lambda & P \\
\hline
 0.8 & (\boldsymbol{5},2,2) & 5 & 24 & 16.3636 & 0.808263 \\
 0.85 & (\boldsymbol{6},2,2) & 5 & 30 & 19.6364 & 0.897849 \\
 0.9 & (\boldsymbol{7},2,2) & 5 & 36 & 22.9091 & 0.948655 \\
 0.95 & (\boldsymbol{8},2,2) & 5 & 42 & 26.1818 & 0.97543 \\
\hline
\end{array}

\\
\\

\multicolumn{2}{l}{\bullet\ \text{Model $(A \times \boldsymbol{C}) \succ \boldsymbol{B}$,
pivot $\boldsymbol{c}$}}

\\

\begin{array}[t]{lllll}
\hline
(b,\boldsymbol{c},n) & df_1 & df_2 & \lambda & P \\
\hline
 (2,\boldsymbol{2},6) & 5 & 5 & 8. & 0.241845 \\
 (3,\boldsymbol{2},4) & 5 & 5 & 9.3913 & 0.278819 \\
 (4,\boldsymbol{2},3) & 5 & 5 & 10.2857 & 0.302586 \\
 (6,\boldsymbol{2},2) & 5 & 5 & 11.3684 & 0.331214 \\
 (2,\boldsymbol{3},4) & 5 & 10 & 11.3684 & 0.4915 \\
 (4,\boldsymbol{3},2) & 5 & 10 & 14.4 & 0.602299 \\
 (2,\boldsymbol{4},3) & 5 & 15 & 14.4 & 0.684104 \\
 (3,\boldsymbol{4},2) & 5 & 15 & 16.6154 & 0.754655 \\
 (2,\boldsymbol{6},2) & 5 & 25 & 19.6364 & 0.885509 \\
\hline
\end{array}
&
\begin{array}[t]{llllll}
\hline
P_{\text{requ.}} & (b,\boldsymbol{c},n) & df_1 & df_2 & \lambda & P \\
\hline
 0.8 & (2,\boldsymbol{6},2) & 5 & 25 & 19.6364 & 0.885509 \\
 0.85 & (2,\boldsymbol{6},2) & 5 & 25 & 19.6364 & 0.885509 \\
 0.9 & (2,\boldsymbol{7},2) & 5 & 30 & 22.9091 & 0.941747 \\
 0.95 & (2,\boldsymbol{8},2) & 5 & 35 & 26.1818 & 0.971837 \\
\hline
\end{array}

\\
\\

\multicolumn{2}{l}{\bullet\ \text{Model $\boldsymbol{V} \succ A$,
pivot $\boldsymbol{n}$}} \\

\begin{array}[t]{lllll}
\hline
(v,\boldsymbol{n}) & df_1 & df_2 & \lambda & P \\
\hline
 (6,\boldsymbol{2}) & 30 & 36 & 4.8 & 0.109714 \\
 (4,\boldsymbol{3}) & 20 & 48 & 7.2 & 0.210406 \\
 (3,\boldsymbol{4}) & 15 & 54 & 9.6 & 0.351949 \\
 (2,\boldsymbol{6}) & 10 & 60 & 14.4 & 0.659852 \\
\hline
\end{array}
&
\begin{array}[t]{llllll}
\hline
P_{\text{requ.}} & (v,\boldsymbol{n}) & df_1 & df_2 & \lambda & P \\
\hline
 0.8 & (2,\boldsymbol{8}) & 10 & 84 & 19.2 & 0.829324 \\
 0.85 & (2,\boldsymbol{9}) & 10 & 96 & 21.6 & 0.884471 \\
 0.9 & (2,\boldsymbol{10}) & 10 & 108 & 24. & 0.923847 \\
 0.95 & (2,\boldsymbol{11}) & 10 & 120 & 26.4 & 0.951 \\
\hline
\end{array}
\end{array}
\end{equation*}
\end{table}

The next example illustrates the minimal sample size computation
for ANOVA models, based on our main results.

\begin{example}  \label{exam:corr}
\begin{enumerateroman}
\item  \label{item:twofact}
Consider the model $A \times \boldsymbol{B}$.
Let $\alpha=0.05$, let $a=6$, let $\delta = \sigma_y$
and consider the power requirement $P\geq 0.9$.
From \Cref{thm:pivot} we observe
that the minimal design has $n=2$
and only the ``pivot'' parameter $b$ is relevant.
By \Cref{thm:main} and \Cref{thm:ineq} we obtain
that to achieve $P\geq 0.9$, the minimal design is $(b,n)=(35,2)$,
with size $abcn=420$ and power $P = 0.909083$.

\item  \label{item:threefact}
Consider the model
$A \times \boldsymbol{B} \times \boldsymbol{C}$ and assume
$\sigma_{\alpha\gamma}^2=0$.
This model is equivalent to
the exact $F$\mbox-test models $(A \times \boldsymbol{B}) \succ \boldsymbol{C}$
and $A \times (\boldsymbol{B} \succ \boldsymbol{C})$,
cf.~\Cref{lemm:special} below.
Let $\alpha=0.05$, let $a=6$, let $\delta = \sigma_y$
and consider the power requirement $P\geq 0.9$.
By \Cref{thm:pivot} we obtain
that the minimal design has $c=n=2$
and only the ``pivot'' parameter $b$ is relevant.
By \Cref{thm:main} and \Cref{thm:ineq} we obtain
that to achieve $P\geq 0.9$,
the minimal design is $(b,c,n)=(35,2,2)$,
with size $abcn=840$ and power $P = 0.909083$.
\end{enumerateroman}
\end{example}

\begin{remark}
In \Cref{exam:corr} the power $P = 0.909083$ for $(b,c,n)=(35,2,2)$
in \ref{item:threefact} is the same as the power for $(b,n)=(35,2)$
in \ref{item:twofact}.
This coincidence is implied by the fact that
\ref{item:twofact} and \ref{item:threefact} have the same d.f.\ and
in the worst case of~\ref{item:twofact} and \ref{item:threefact}
the total variance is consumed entirely by $\sigma_{\alpha\beta}^2$,
cf.\ \Cref{rema:ineq}\ref{item:leastfavorable}.
\end{remark}

\section{Models with approximate \texorpdfstring{$\boldsymbol{F}$\mbox-test}{F-test}}  \label{sec:approx}

For the two models
\begin{equation}  \label{eq:approx}
A \times \boldsymbol{B} \times \boldsymbol{C}
\qquad\text{and}\qquad
(A \succ \boldsymbol{B}) \times \boldsymbol{C},
\end{equation}
an exact $F$\mbox-test does not exist.
Approximate $F$\mbox-tests can be obtained
by Satterthwaite's approximation
that goes back to \citet{BEHRENS1929}, \citet{WELCH1938}, \citet{WELCH1947}
and generalized by \citet{SATTERTHWAITE1946},
see \citet[Appendix\,K]{SAHAI2000}.
The details of the approximate $F$\mbox-tests for the models in \eqref{eq:approx}
are in \citet[Sec.\,3.4.1.3 and Sec.\,3.4.4.5]{RASCH2011}.
Satterthwaite's approximation in a similar or different form
also occurs, for example, in
\citet{DAVENPORT1972},
\citet{DAVENPORT1973},
\citet[pp.\,40--41]{DONCASTER2007},
\citet{HUDSON1968},
\citet{LORENZEN2019},
\citet{RASCH2012},
\citet{WANG2005},
also denoted as quasi-$F$\mbox-test \citep{MYERS2010}.

The approximate $F$-test d.f.\ involve mean squares to be simulated.
To approximate the power of the test,
simulate data such that $H_0$ is false and compute the rate of rejections.
The rate approximates the power of the test.
In the middle plot of \Cref{fig:pwr} we give an example of the power behavior
for the approximate $F$\mbox-test model $(A \succ \boldsymbol{B}) \times \boldsymbol{C}$.
The plot shows that the ``pivot'' effect for exact $F$\mbox-tests (\Cref{thm:pivot})
does not generalize to approximate $F$\mbox-tests.

The next lemma rephrases observations in \citet{RASCH2011,RASCH2012}.
It allowed us to avoid approximations but use exact $F$\mbox-test computations for
the left and the right plots of \Cref{fig:pwr}.

\begin{lemma}  \label{lemm:special}
The following special cases of~\eqref{eq:approx}
are equivalent to exact $F$\mbox-test models,
in the sense of identical d.f.\ and noncentrality parameters.
\begin{enumerateroman}
\item
If in the model $A \times \boldsymbol{B} \times \boldsymbol{C}$
we have $\sigma_{\alpha\gamma}^2=0$, then it is equivalent to
$(A \times \boldsymbol{B}) \succ \boldsymbol{C}$
and $A \times (\boldsymbol{B} \succ \boldsymbol{C})$.

\item  \label{item:mixed}
If in the model $(A \succ \boldsymbol{B}) \times \boldsymbol{C}$
we have $\sigma_{\beta(\alpha)}^2=0$, then it is equivalent to
$(A \times \boldsymbol{C}) \succ \boldsymbol{B}$
and $A \times (\boldsymbol{C} \succ \boldsymbol{B})$;
while if $\sigma_{\alpha\gamma}^2=0$, then it is equivalent
to $A \succ \boldsymbol{B} \succ \boldsymbol{C}$.
\end{enumerateroman}
\end{lemma}

\begin{proof}
The equivalences follow from inspecting the d.f.\ and the noncentrality parameter.
\end{proof}

\begin{remark}
To look up in \Cref{tab:main}
the first case of \Cref{lemm:special}\ref{item:mixed},
swap the factor names $B\leftrightarrow C$ first.
\end{remark}

\begin{figure}[!htbp]
\centering
\includegraphics[width=\textwidth]{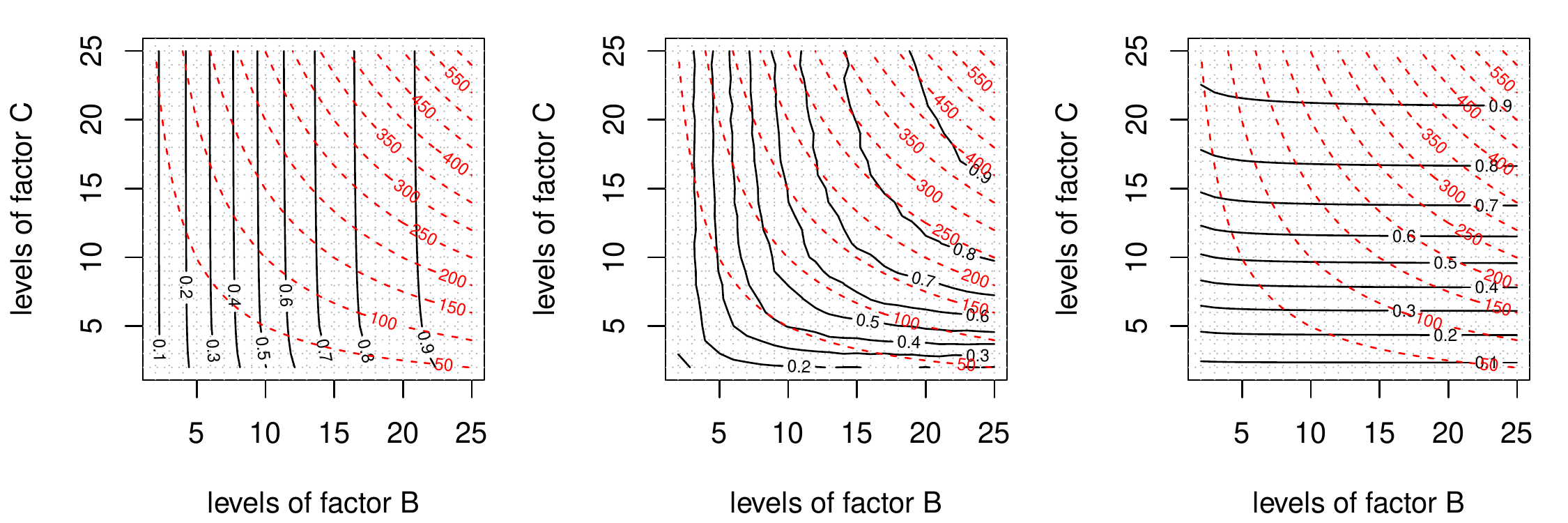}
\caption{
Power and size for the mixed model $(A \succ \boldsymbol{B}) \times \boldsymbol{C}$,
for $a=6$, $\alpha=0.05$, $\delta=5$, and
three variance component assignments
$(\sigma_{\beta(\alpha)}^2,\ \sigma_{\gamma}^2,\ 
\sigma_{\alpha\gamma}^2,\ \sigma_{\beta\gamma(\alpha)}^2,\ \sigma^2)$
$= (10,5,0,5,5)$,
$(5,5,5,5,5)$,
$(0,5,10,5,5)$, from left to right.
Each contour plots shows the guaranteed power
$P_{\min} = (1-\beta)_{\min}$ (solid curves) overlaid
with the size factor $b\cdot c$ (red, dashed hyperbolas)
as functions of $b,c\leq 25$, for fixed $n=2$.
By \Cref{lemm:special}\ref{item:mixed} the left model is equivalent
to $A \succ \boldsymbol{B} \succ \boldsymbol{C}$,
such that by \Cref{thm:pivot} the ``pivot'' parameter is $b$.
The middle plot is an
approximate $F$\mbox-test model
(the power is approximated by 10\,000 simulations),
there is no ``pivot'' effect.
The right model is equivalent
to $(A \times \boldsymbol{C}) \succ \boldsymbol{B}$,
the ``pivot'' parameter is $c$.
}
\label{fig:pwr}
\end{figure}

\section{Real versus integer parameters}  \label{sec:integer}

The ``pivot'' effect for the minimum sample size
described in \Cref{thm:pivot}\ref{item:size}
is formulated with the assumption that the parameters are real numbers.
The effect also occurs in most practical examples,
where the parameters are integers.
But we constructed the following example to point out
that for integer parameters the ``pivot'' effect is not a granted fact.

\begin{example}  \label{exam:integer}
Consider the two-way model $A \times \boldsymbol{B}$
with $a = 15$, $\alpha = 0.1$, $\delta = 7$,
$( \sigma_{\beta(\alpha)}^2, \sigma^2 ) = (0.01, 8)$,
and required power $P \geq 0.9$.
Then for real $b,n \geq 2$,
the minimum sample size obtained by \Cref{thm:pivot}\ref{item:size}
occurs for $(b,n)=(4.019937,2)$,
where $P=0.9$.
For integers $b,n = 2,3,\dots$,
the minimum sample size occurs for $(b,n)=(3,3)$,
where $P=0.902873$.
Thus in this example the ``pivot'' effect
is obstructed if we switch from real numbers to integers.
In more realistic examples this obstruction does not occur.
\end{example}

\begin{remark}
While \Cref{exam:integer} shows that the transition to integers
can obstruct (if by an unrealistic example) the ``pivot'' effect,
we remark that the obstruction is limited,
that is, the real number computation has the following valid implication
for the integer result.
The real number minimum at $(b,n)=(4.019937,2)$,
readily computed by using \Cref{thm:pivot}\ref{item:size},
immediately implies that the integer minimum size occurs at $(b,n)$
with $b\cdot n$ between $4.019937\cdot 2$ and $5\cdot 2$, that is,
\begin{equation*}
b\cdot n \in \{9,10\},
\end{equation*}
in fact in the example $b\cdot n = 9$.
A similar implication holds for all models in \Cref{tab:main}.
\end{remark}

\section{Conclusions}

We determine the noncentrality parameter of the exact $F$\mbox-test
for balanced factorial ANOVA models.
From a sharp lower bound for the noncentrality parameter
we obtain the power that can be guaranteed
in the least favorable case.
These results allow us to compute the minimal sample size,
but we also provide a structural result for the minimal sample size.
The structural result is formulated as a ``pivot'' effect,
which means that one of the factors is more relevant than
the others, for the power and thus for the minimum sample size.

\section*{Acknowledgments}

The authors are grateful to Karl Moder for helpful discussions and comments.
We also thank the reviewer for useful comments.

\section*{ORCID}

Bernhard Spangl  \myorcidlong{https://orcid.org/0000-0001-6222-2408} \\
Norbert Kaiblinger  \myorcidlong{https://orcid.org/0000-0001-6280-5929} \\
Peter Ruckdeschel  \myorcidlong{https://orcid.org/0000-0001-7815-4809} \\
Dieter Rasch  \myorcidlong{https://orcid.org/0000-0001-9324-9910}

\numberwithin{equation}{section}

\renewcommand{\thesection}{Appendix A}
\section{Proofs}  \label{sec:proofs}
\renewcommand{\thesection}{A}

We include a short proof of the formula
for the noncentrality parameter in \citet[p.\,151]{LINDMAN1992},
formulated here in a more general form.

\begin{lemma}  \label{lemm:noncent}
Let a test statistic $\boldsymbol{F}$ have
a noncentral $F$\mbox-distribution
with numerator and denominator d.f.\ $df_1$ and $df_2$, respectively,
written as $\boldsymbol{F} = \boldsymbol{Z}_1 / \boldsymbol{Z}_2$, with $q\neq0$,
\begin{equation*}
\begin{aligned}
\boldsymbol{Z}_1 & = q \cdot \boldsymbol{X}_1/df_1, \\
\boldsymbol{Z}_2 & = q \cdot \boldsymbol{X}_2/df_2,
\end{aligned}
\qquad
\begin{aligned}
\boldsymbol{X}_1 & \sim \chi^2(df_1,\lambda), \\
\boldsymbol{X}_2 & \sim \chi^2(df_2,0),
\end{aligned}
\end{equation*}
and $\boldsymbol{X}_1,\boldsymbol{X}_2$ stochastically independent.
Then the noncentrality parameter $\lambda$ satisfies
\begin{equation*}
\lambda = \mathit{df_1} \cdot \left ( \frac{E(\boldsymbol{Z}_1)}{E(\boldsymbol{Z}_2)} - 1 \right ) .
\end{equation*}
\end{lemma}

\begin{proof}
Since $E(\boldsymbol{X}_1) = df_1 + \lambda$ and $E(\boldsymbol{X}_2) = df_2$,
we obtain $E(\boldsymbol{Z}_1) = q \cdot (1 + \lambda / df_1)$ and $E(\boldsymbol{Z}_2) = q$.
Hence,
\begin{equation}
E(\boldsymbol{Z}_1) / E(\boldsymbol{Z}_2)= 1 + \lambda / df_1,
\end{equation}
which implies the expression for $\lambda$ in the lemma.
\end{proof}

\begin{remark}
Jensen's equality implies
$E(\boldsymbol{Z}_1)/E(\boldsymbol{Z}_2) < E(\boldsymbol{Z}_1/\boldsymbol{Z}_2) = E(\boldsymbol{F})$.
For $E(\boldsymbol{F})$, see \citet[formula (30.3a)]{JOHNSON1995}.
\end{remark}

The next lemma summarizes
monotonicity properties of the noncentral $F$-distribution
from \citet{GHOSH1973},
listed in \citet[Sec.\,16.4.2]{HOCKING2003},
see also \citet[Theorem 4.3]{FINNER1997} with a sharper statement.
Recall that for $0 \le \gamma \le 1$,
we let $F_{\mathit{df_1}, \mathit{df_2}; \gamma}$ denote the
$\gamma$-quantile of the central $F$-distribution with $\mathit{df_1}$
and $\mathit{df_2}$ degrees of freedom.

\begin{lemma} \label{lemm:monotony}
Let $F$ be distributed according
to the noncentral $F$-distribution $F_{\mathit{df_1}, \mathit{df_2}}^\lambda$
with noncentrality parameter $\lambda$.
Then referring to the probability
$\mathbb{P}(F > F_{\mathit{df_1}, \mathit{df_2}; \gamma})$ as power,
we have if $\mathit{df_1}$ decreases
and $\mathit{df_2}$, $\lambda$ increase, then the power increases.
That is, we have the implication
\begin{equation*}
\left.  
\begin{aligned}
\mathit{df_1} & \ge \mathit{df_1'} \\
\mathit{df_2} & \le \mathit{df_2'} \\
\lambda & \le \lambda'
\end{aligned} \quad \right \}
\qquad \Rightarrow \qquad
  \mathbb{P}(F > F_{\mathit{df_1}, \mathit{df_2}; \gamma}) \le
  \mathbb{P}(F' > F_{\mathit{df_1'}, \mathit{df_2'}; \gamma})
\end{equation*}
with $F \sim F_{\mathit{df_1}, \mathit{df_2}}^\lambda$ and
$F' \sim F_{\mathit{df_1'}, \mathit{df_2'}}^{\lambda'}$~.
\end{lemma}

\begin{proof}
For varying $\mathit{df_1}$, see \citet[Thm.\,6]{GHOSH1973}.
For varying $\mathit{df_2}$, apply \citet[Thm.\,5]{GHOSH1973} with $\lambda_0=0$.
For varying $\lambda$,
see \citet[p.\,219, Satz~2.36(b)]{WITTING1985}
or \citet[p.\,53, Exercise~2.9]{BHATTACHARYA2016}.
\end{proof}

\begin{proof}[\textbf{Proof of {\upshape\Cref{thm:main}}}]
We prove the result only for the model
$A \succ \boldsymbol{B} \succ \boldsymbol{C}$,
the proofs for the other models are analogous.
In the expected mean squares table \citep[p.\,100, Table~3.15]{RASCH2011}
the two expressions
\begin{equation}  \label{eq:numdenom}
\begin{aligned}
E(\MS_A)
& = \sigma^2 + n \sigma_{\gamma(\alpha\beta)}^2 + cn \sigma_{\beta(\alpha)}^2
+ \frac{bcn}{a-1} \myssa \\
E(\MS_{B \, \text{in} \, A})
& = \sigma^2 + n \sigma_{\gamma(\alpha\beta)}^2 + cn \sigma_{\beta(\alpha)}^2
\end{aligned}
\end{equation}
are equal under the null hypothesis $H_0$ of no $A$\mbox-effects.
Hence, $H_0$ can be tested by the exact $F$\mbox-test
\begin{equation}  \label{eq:teststat}
\boldsymbol{F}_{A} = \frac{\MS_{A}}{\MS_{B \, \text{in} \, A}} ~,
\end{equation}
which under $H_{0}$ is central $F$\mbox-distributed,
in general noncentral $F$-distributed.
From the ANOVA table \citep[p.\,91, Table~3.10]{RASCH2011}
the numerator and denominator d.f.\ are
$\mathit{df_1} = a-1$ and $\mathit{df_2} = a (b-1)$, respectively.
By \Cref{lemm:noncent} the noncentrality parameter $\lambda$ is thus
\begin{equation}  \label{eq:lam}
\begin{aligned}
\lambda
& = \frac{bcn \myssa}{\sigma^2 + n \sigma_{\gamma(\alpha\beta)}^2 + cn \sigma_{\beta(\alpha)}^2}
= b \cdot \frac{\myssa}{\sigma_{\beta(\alpha)}^2 + \frac{1}{c} \sigma_{\gamma(\alpha\beta)}^2 + \frac{1}{cn}\sigma^2} ~.
\end{aligned}
\end{equation}
\end{proof}

\begin{remark}  \label{rema:denominator}
\begin{enumerateroman}
\item
The formula \eqref{eq:lam} allows us to point out
the difference of our results
compared to the previous literature
\citep[p.58--59]{RASCH2011}.
In fact, the expression $bcn$ in the numerator at the left-hand side of \eqref{eq:lam}
coincides with the expression $C$ in \citet[Table~3.2]{RASCH2011},
but note that the denominator is distinct.
The exact expression for $\lambda$ in \eqref{eq:lam} has the sum of variance components
$\sigma_y^2 = \sigma^2 + \sigma_{\gamma(\alpha\beta)}^2 + \sigma_{\beta(\alpha)}^2$
replaced by the linear combination
$\sigma^2 + n \sigma_{\gamma(\alpha\beta)}^2 + cn \sigma_{\beta(\alpha)}^2$,
see also \citet{RASCH2020a}.
The fourth author and Rob Verdooren have acknowledged our results
and update their available R-programs accordingly,
note in \citet{RASCH2020a} some citation numbers have been mixed up.
To reproduce the examples of the present paper,
R-code is available from the first author.

\item
The transformation from the left-hand side to the right-hand side in \eqref{eq:lam}
shifts the attention from the product of parameters $bcn$ to the single parameter $b$.
This observation is the key to our general ``pivot'' effect result (\Cref{thm:pivot}).

\item
To verify the details of \Cref{tab:main}
note that the expected mean squares table entries used in the proof of \Cref{thm:main} depend on the factors being fixed or random.
\end{enumerateroman}
\end{remark}

\begin{proof}[\textbf{Proof of {\upshape\Cref{thm:ineq}}}]
\begin{enumerateroman}
\item
As above we prove the result for the model $A \succ \boldsymbol{B} \succ \boldsymbol{C}$.
Since
\begin{equation}
\sigma_{\beta(\alpha)}^2 + \frac{1}{c} \sigma_{\gamma(\alpha\beta)}^2 + \frac{1}{cn}\sigma^2
\leq \underbrace{\sigma_{\beta(\alpha)}^2 + \sigma_{\gamma(\alpha\beta)}^2 + \sigma^2}_
{\displaystyle\sigma_{y,\mathit{active}}^2},
\end{equation}
we obtain
\begin{equation}
\lambda
= b \cdot \frac{\myssa}{\sigma_{\beta(\alpha)}^2 + \frac{1}{c} \sigma_{\gamma(\alpha\beta)}^2 + \frac{1}{cn}\sigma^2}
\geq  b \cdot \frac{\myssa}{\sigma_{y,\mathit{active}}^2},
\end{equation}
and the Sz\H{o}kefalvi-Nagy inequality
\citep{ALPARGU2000,BRAUER1959,GUTMAN2017,KAIBLINGER2020,SHARMA2010,SZOKEFALVINAGY1918}
states that
\begin{equation}  \label{eq:szoek}
\myssa \geq \frac{(\alpha_{\max} - \alpha_{\min})^2}{2}
= \frac{\delta^2}{2} ~.
\end{equation}

\item[(ii),(iii)]
By \citet{KAIBLINGER2020} we have
for the $(v\times a)$ matrix $(\alpha_{i(j)})_{i,j}$
and for the $(u \times v \times a)$ array $(\alpha_{i(jk)})_{i,j,k}$,
\begin{equation}
\myssaa \geq
\frac{\delta^2}{2} \cdot \frac{m}{m-1}
\quad\text{ and }\quad
\myssaaa \geq
\frac{\delta^2}{2} \cdot \frac{m_2 m_3}{(m_2-1)(m_3-1)},
\end{equation}
respectively.  \qedhere
\end{enumerateroman}
\end{proof}

\begin{proof}[\textbf{Proof of {\upshape\Cref{thm:pivot}}}]
\begin{enumerateroman}
\item
We consider the parameters as competitors in
\begin{equation}  \label{eq:compet}
\text{not increasing $\mathit{df_1}$ \qquad and \qquad increasing $\mathit{df_2}$ and $\lambda$.}
\end{equation}
For each model in \Cref{tab:main},
we analyze the effect of the parameters on $\mathit{df_1}$, $\mathit{df_2}$ and $\lambda$,
using the arguments illustrated in \Cref{exam:pivotproof} below.
The inspection yields that for each model there is a sole winner,
which we call the ``pivot'' parameter.
We exemplify the scoring for four models:
\begin{center}
\begin{tabular}{l|llll}
& $A \succ B \succ C$
& $A \succ \boldsymbol{B} \succ \boldsymbol{C}$
& $V \succ A \succ B$
& $V \succ A \succ \boldsymbol{B}$ \\
\hline
parameters
& $b,c,n$
& $b,c,n$
& $v,b,n$
& $v,b,n$ \\
\hline
least increase in $\mathit{df_1}$ & $b,c,n$ & $b,c,n$ & $b,n$ & $b,n$ \\
most increase in $\mathit{df_2}$ & $n$ & $b$ & $n$ & $b$ \\
most increase in $\lambda$ & $b,c,n$ & $b$ & $b,n$ & $b$ \\
\hline
$\Rightarrow$ pivot & $n$ & $b$ & $n$ & $b$
\end{tabular}
\end{center}
Since by \Cref{lemm:monotony} the lead in \eqref{eq:compet} also means the lead in power increase, we thus obtain that the ``pivot'' yields the maximal power increase.

\item
Start with minimal parameters and apply \ref{item:pivotone}.

\item is equivalent to \ref{item:power}.  \qedhere
\end{enumerateroman}
\end{proof}

\begin{example}
We illustrate the proof of \Cref{thm:pivot}\ref{item:pivotone}
by showing the typical argument for most increase in $\mathit{df_2}$
and the typical argument for most increase in $\lambda$.
\begin{enumerateroman}  \label{exam:pivotproof}
\item
In the model $A \succ B \succ C$
the parameter $n$ is more effective than $b$ or $c$ in increasing $\mathit{df_2}$,
\begin{equation}  \label{eq:dftwo}
\mathit{df_2} = abc(n-1) = abcn - abc,
\end{equation}
since $b,c,n$ equally increase the positive term of \eqref{eq:dftwo},
but only $n$ does not increase the negative term.

\item
For the model $A \succ \boldsymbol{B} \succ \boldsymbol{C}$,
the parameter $b$ is more effective than $c$ or $n$ in increasing $\lambda$,
\begin{equation}  \label{eq:lambda}
\lambda = \frac{bcn \myssa}{\sigma^2 + n \sigma_{\gamma(\alpha\beta)}^2 + cn \sigma_{\beta(\alpha)}^2},
\end{equation}
since $b,c,n$ equally increase the numerator of \eqref{eq:lambda},
but only $b$ does not increase the denominator.
\end{enumerateroman}
\end{example}


\begin{thebibliography}{99}

\bibitem[Alpargu and Styan(2000, p.\,11)]{ALPARGU2000}
Alpargu, G., Styan, G.P.H.,
2000.
Some comments and a bibliography on the Frucht-Kantorovich and Wielandt inequalities,
in: Innovations in Multivariate Statistical Analysis, Springer, 1--38.
\doi{10.1007/978-1-4615-4603-0}

\bibitem[Behrens(1929)]{BEHRENS1929}
Behrens, W.V.,
1929.
Ein Beitrag zur Fehlerberechnung bei wenigen Beobachtungen.
Landw. Jahrb., 68, 807--837.
(German).

\bibitem[Bhattacharya and Burman(2016)]{BHATTACHARYA2016}
Bhattacharya, P.K., Burman, P.,
2016.
Theory and Methods of Statistics.
Elsevier.
\href{https://www.elsevier.com/books/theory-and-methods-of-statistics/bhattacharya/978-0-12-802440-9}{ISBN 9780128041239}

\bibitem[Brauer and Mewborn(1959)]{BRAUER1959}
Brauer, A., Mewborn, A.C.,
1959.
The greatest distance between two characteristic roots of a matrix.
Duke Math. J. 26 (4), 653--661.
\doi{10.1215/S0012-7094-59-02663-8}

\bibitem[Canavos and Koutrouvelis(2009)]{CANAVOS2009}
Canavos, G., Koutrouvelis, I.,
2009.
An Introduction to the Design \& Analysis of Experiments.
Pearson.
\href{https://www.pearson.com/us/higher-education/program/Canavos-Introduction-to-the-Design-Analysis-of-Experiments/PGM295527.html}{ISBN 978-0136158639}

\bibitem[Davenport and Webster(1972)]{DAVENPORT1972}
Davenport, J.M., Webster, J.T.,
1972.
Type-I Error and Power of a Test Involving a Satterthwaite's Approximate F-Statistic.
Technometrics 14 (3), 555--569.

\bibitem[Davenport and Webster(1973)]{DAVENPORT1973}
Davenport, J.M., Webster, J.T.,
1973.
A Comparison of Some Approximate F-Tests.
Technometrics 15 (4), 779--789.

\bibitem[Doncaster and Davey(2007)]{DONCASTER2007}
Doncaster, C.P., Davey, A.J.H.,
2007.
Analysis of Variance and Covariance:
How to Choose and Construct Models for the Life Sciences.
Cambridge Univ. Press.
\doi{10.1017/CBO9780511611377}

\bibitem[Finner and Roters(1997)]{FINNER1997}
Finner, H., Roters, M.,
1997.
Log-concavity and inequalities for chi\mbox-square, F and beta distributions
with applications in multiple comparisons.
Stat. Sinica 7 (3), 771--787.

\bibitem[Fox(2015)]{FOX2015}
Fox, J.,
2015.
Applied Regression Analysis and Generalized Linear Models. (3rd ed.)
SAGE Publ.
\href{https://us.sagepub.com/en-us/nam/applied-regression-analysis-and-generalized-linear-models/book237254}{ISBN 9781452205663}

\bibitem[Ghosh(1973)]{GHOSH1973}
Ghosh, B.K.,
1973.
Some monotonicity theorems for $\chi^2$, $F$ and $t$ distributions
with applications.
J. R. Stat. Soc., Ser. B 35 (3), 480--492.

\bibitem[Gutman, Das, Furtula, Milovanovi\'c,
and Milovanovi\'c(2017)]{GUTMAN2017}
Gutman, I., Das, K.C., Furtula, B., Milovanovi\'c, E., Milovanovi\'c, I.,
2017.
Generalizations of Sz\H{o}kefalvi Nagy and Chebyshev inequalities
with applications in spectral graph theory.
Appl. Math. Comput. 313, 235--244.
\doi{10.1016/j.amc.2017.05.064}

\bibitem[Hocking(2003)]{HOCKING2003}
Hocking, R.R.,
2003.
Methods and Applications of Linear Models: Regression and the Analysis of Variance.
Wiley.
\doi{10.1002/0471434159}

\bibitem[Hudson and Krutchkoff(1968)]{HUDSON1968}
Hudson, J.D. Jr, Krutchkoff, R.G.,
1968.
A Monte Carlo investigation of the size and power of tests employing
Satterthwaite's synthetic mean squares.
Biometrika 55 (2), 431--433.

\bibitem[Jiang(2007)]{JIANG2007}
Jiang, J.,
2007.
Linear and Generalized Linear Mixed Models and Their Applications.
Springer.
\doi{10.1007/978-0-387-47946-0}

\bibitem[Johnson, Kotz, and Balakrishnan(1995)]{JOHNSON1995}
Johnson, N.L., Kotz, S., Balakrishnan, N.,
1995.
Continuous Univariate Distributions, Volume 2. (2nd ed.)
\href{https://www.wiley.com/en-us/Continuous+Univariate+Distributions\%2C+Volume+2\%2C+2nd+Edition-p-9780471584940}{ISBN 9780471584940}

\bibitem[Kaiblinger and Spangl(2020)]{KAIBLINGER2020}
Kaiblinger, N., Spangl, B.,
2020.
An inequality for the analysis of variance.
Math. Inequal. Appl. 23 (3), 961--969.
\doi{10.7153/mia-2020-23-74}

\bibitem[Lindman(1992)]{LINDMAN1992}
Lindman, H.R.,
1992.
Analysis of Variance in Experimental Design.
Springer.
\doi{10.1007/978-1-4613-9722-9}

\bibitem[Lorenzen and Anderson(2019)]{LORENZEN2019}
Lorenzen, T., Anderson, V.,
2019.
Design of Experiments: A No-Name Approach.
CRC Press.
\href{https://www.crcpress.com/Design-of-Experiments-A-No-Name-Approach/Lorenzen-Anderson/p/book/9780367402327}{ISBN 9780367402327}

\bibitem[Maxwell, Delaney, and Kelley(2017)]{MAXWELL2017}
Maxwell, S.E., Delaney, H.D., Kelley, K.,
2017.
Designing Experiments and Analyzing Data.
Routledge.
\href{https://www.routledge.com/Designing-Experiments-and-Analyzing-Data-A-Model-Comparison-Perspective/Maxwell-Delaney-Kelley/p/book/9781138892286}{ISBN 9781138892286}

\bibitem[Montgomery(2017)]{MONTGOMERY2012}
Montgomery, D.,
2017.
Design and Analysis of Experiments.
Wiley.
\href{https://www.wiley.com/Design+and+Analysis+of+Experiments\%2C+9th+Edition-p-9781119320937}{ISBN 978-1-119-32093-7}

\bibitem[Myers(2010)]{MYERS2010}
Myers, J.L., Well, A.D.,
2010.
Research Design and Statistical Analysis.
Taylor \& Francis.
\doi{10.4324/9780203726631}

\bibitem[Rasch(1971)]{RASCH1971}
Rasch, D.,
1971.
Gemischte Klassifikation der dreifachen Varianzanalyse.
Biometr. Z. 13 (1), 1--20.
(German).
\doi{10.1002/bimj.19710130102}

\bibitem[Rasch, Pilz, Verdooren, and Gebhardt(2011)]{RASCH2011}
Rasch, D., Pilz, J., Verdooren, R., Gebhardt, A.,
2011.
Optimal Experimental Design with R.
Chapman \& Hall.
\doi{10.1007/s00362-012-0473-y}

\bibitem[Rasch and Schott(2018)]{RASCH2018}
Rasch, D., Schott, D.,
2018.
Mathematical Statistics.
Wiley.
\doi{10.1002/9781119385295}

\bibitem[Rasch, Spangl, and Wang(2012)]{RASCH2012}
Rasch, D., Spangl, B., Wang, M.,
2012.
Minimal experimental size in the three way ANOVA cross
classification model with approximate $F$\mbox-tests.
Commun. Stat., Simulation Comput. 41 (7), 1120--1130.
\doi{10.1080/03610918.2012.625832}

\bibitem[Rasch and Verdooren(2020)]{RASCH2020a}
Rasch, D., Verdooren, R.,
2020.
Determination of minimum and maximum experimental size in
one-, two- and three-way ANOVA with fixed and mixed models by R.
J. Stat. Theory Pract. 14 (4), \#57.
\doi{10.1007/s42519-020-00088-6}

\bibitem[Rasch, Verdooren, and Pilz(2020)]{RASCH2020b}
Rasch, D., Verdooren, R., Pilz, J.,
2020.
Applied Statistics.
Wiley.
\doi{10.1002/9781119551584}

\bibitem[Sahai and Ageel(2000)]{SAHAI2000}
Sahai, H., Ageel, M.I.,
2000.
The Analysis of Variance, Birkh{\"a}user.
\doi{10.1007/978-1-4612-1344-4}

\bibitem[Satterthwaite(1946)]{SATTERTHWAITE1946}
Satterthwaite, F.,
1946.
An approximate distribution of estimates of the variance components.
Biometrics Bull. 2 (6), 110--114.
\doi{10.2307/3002019}

\bibitem[Scheff\'e(1959)]{SCHEFFE1959}
Scheff\'e, H.,
1959.
The Analysis of Variance.
Wiley.
\href{https://www.wiley.com/en-us/The+Analysis+of+Variance-p-9780471345053}{ISBN 9780471345053}

\bibitem[Searle and Gruber(2017)]{SEARLE2016}
Searle, S.R., Gruber, M.H.J.,
2017.
Linear Models. (2nd ed.)
Wiley.
\href{https://www.wiley.com/en-us/Linear+Models\%2C+2nd+Edition-p-9781118952832}{ISBN 9781118952856}

\bibitem[Sharma, Gupta, and Kapoor(2010)]{SHARMA2010}
Sharma, R., Gupta, M., Kapoor, G.,
Some better bounds on the variance with applications.
J. Math. Inequal. 4(3), 355--363.
\doi{10.7153/jmi-04-32}

\bibitem[Sz{\H o}kefalvi-Nagy(1918)]{SZOKEFALVINAGY1918}
Sz\H{o}kefalvi-Nagy, J.,
1918.
{\"U}ber algebraische Gleichungen mit lauter reellen Wurzeln.
Jahresber. Dtsch. Math.-Ver. 27, 37--43.
(German).

\bibitem[Wang, Rasch, and Verdooren(2005)]{WANG2005}
Wang, M., Rasch, D., Verdooren, R.,
2005.
Determination of the size of a balanced experiment
in mixed ANOVA models using the modified approximate F-test.
J. Stat. Plann. Inference 132, 183--201.
\doi{10.1016/j.jspi.2004.06.022}

\bibitem[Welch(1938)]{WELCH1938}
Welch, B.L.,
The significance of the difference between two means when the
population variances are unequal.
1938.
Biometrika 29 (3-4), 350--362.
\doi{10.1093/biomet/29.3-4.350}

\bibitem[Welch(1947)]{WELCH1947}
Welch, B.L.,
1947.
The generalization of `Student's' problem when several different
population variances are involved.
Biometrika 34 (1-2), 28--35.
\doi{10.1093/biomet/34.1-2.28}

\bibitem[Witting(1985)]{WITTING1985}
Witting, H.,
1985.
Mathematische Statistik I.
Springer.
(German).
\doi{10.1007/978-3-322-90150-7}

\end{thebibliography}
\end{document}